\documentclass[journal]{IEEEtran}
\usepackage[cmex10]{amsmath}
\usepackage{eqparbox}
\makeatletter
\newcommand{\Bbar}{{\mathchoice
		{\smash@bar\textfont\displaystyle{0.509}{6.2} \mathrm{B}}
		{\smash@bar\textfont\textstyle{0.99}{5.2} B}
		{\smash@bar\scriptfont\scriptstyle{0.25}{1.2} B}
		{\smash@bar\scriptscriptfont\scriptscriptstyle{0.25}{1.2} B}
}}
\newcommand{\Abar}{{\mathchoice
		{\smash@bar\textfont\displaystyle{0.509}{6.2} \mathrm{A}}
		{\smash@bar\textfont\textstyle{0.99}{5.2} A}
		{\smash@bar\scriptfont\scriptstyle{0.25}{1.2} A}
		{\smash@bar\scriptscriptfont\scriptscriptstyle{0.25}{1.2} A}
}}
\newcommand{\smash@bar}[4]{%
	\smash{\rlap{\raisebox{-#3\fontdimen5#10}{$\m@th#2\mkern#4mu\mathchar'26$}}}%
}
\makeatother
\usepackage[usenames, dvipsnames]{color}
\usepackage{mathtools,amssymb,bm,mathabx}
\usepackage{amstext}
\usepackage{amssymb}
\usepackage{graphicx}
\usepackage{color}
\usepackage{booktabs}
\usepackage{longtable}
\usepackage{multicol}
\usepackage{amsfonts}
\usepackage{dsfont}
\usepackage{array}
\usepackage{algorithmicx}
\usepackage[ruled]{algorithm}
\usepackage{algpseudocode}
\usepackage{algpascal}
\usepackage{algc}
\usepackage{cases}
\usepackage{pgfplots}
\usepackage{accents}
\usepackage{caption}
\usepackage[footnotesize]{subfigure}
\usepackage{amsthm,xparse}
\DeclareCaptionLabelSeparator{periodspace}{.\quad}
\usepackage{float}
\usepackage{cite}
\usepackage [autostyle, english = american]{csquotes}
\usepackage{hyperref}
\theoremstyle{remark}

\newcommand\ASTART{\bigskip\noindent\begin{minipage}[b]{0.5\linewidth}}
	
	\newcommand\AENDSKIP{\end{minipage}\bigskip}
\newcommand\AEND{\end{minipage}}
\ifCLASSOPTIONcaptionsoff
\usepackage[nomarkers]{endfloat}
\let\MYoriglatexcaption\caption
\renewcommand{\caption}[2][\relax]{\MYoriglatexcaption[#2]{#2}}
\fi
\theoremstyle{plain}
\newtheorem{thm}{\textbf{Theorem}}

\theoremstyle{definition}

\theoremstyle{remark}

\newcommand*{\rom}[1]{\expandafter\@slowromancap\romannumeral #1@}

\newcommand{\RN}[1]{%
\textup{\uppercase\expandafter{\romannumeral#1}}%
}
\def\change{black}
\def\changes{black}
\def\changess{black}

\usepackage{standalone}
\graphicspath{{./Figures/}} 

\begin{document}
%
\title{Blind Two-Dimensional Super-Resolution in Multiple-Input Single-Output Linear Systems}
\author{Shahedeh~Sayyari, Sajad~Daei, and Farzan~Haddadi
	\thanks{The authors are with the School of Electrical Engineering, Iran University of Science \& Technology.}
}%

\maketitle

\begin{abstract}
	In this paper, we consider a multiple-input single-output (MISO) linear time-varying system whose output is a superposition of scaled and time-frequency shifted versions of inputs. The goal of this paper is to determine system characteristics and input signals from the single output signal. More precisely, we want to recover the continuous time-frequency shift pairs, the corresponding (complex-valued) amplitudes and the input signals from only one output vector. This problem arises in a variety of applications such as radar imaging, microscopy, channel estimation and localization problems. While this problem is naturally ill-posed, by constraining the unknown input waveforms to lie in separate known {\color{\changes}low-dimensional} subspaces, it becomes tractable. More explicitly, we propose a semidefinite program which exactly recovers time-frequency shift pairs and input signals. We prove uniqueness and optimality of the solution to this program. Moreover, we provide a grid-based approach which can significantly reduce computational complexity in exchange for adding a small gridding error. Numerical results confirm the ability of our proposed method to exactly recover the unknowns.	
\end{abstract}
\vspace{-0.2cm}

\IEEEpeerreviewmaketitle

\vspace{-0.2cm}
\section{Introduction}
\IEEEPARstart{S}{uper}-resolution is the problem of recovering {\color{\changes}high-resolution} information from {\color{\changes}low-resolution} data. In this letter, we assume a linear time-varying (LTV) {\color{\changess}multiple-input single-output (MISO)} system in which the inputs are {\color{\changes}continuous-time band-limited} arbitrary signals $x_j(t)$ and the output vector $y(t)$ is a weighted superposition of time and frequency shifted versions of the inputs:
\vspace{-0.1cm}
\begin{align} \label{A}
&y(t)=\sum_{j=0}^{N_{I}-1}\sum_{k=1}^S{b}_k {x}_j(t-\widetilde{\tau}_k)e^{i2\pi\widetilde{\nu}_k {\color{\change} t}}.
\end{align}
Here, ${b}_k\in\mathbb{C}$ are some unknown weights and $\widetilde{\tau}_k  \in [-T/2,T/2]$, $\widetilde{\nu}_k  \in [-W/2,W/2]$ represent unknown continuous time and frequency shifts, respectively. $\widetilde{\tau}_k$ and $\widetilde{\nu}_k$ are not constrained to lie on a predefined domain of grids. The number of input signals $N_I$ is known and we want to recover the unknowns $({b}_k,\widetilde{\tau}_k,\widetilde{\nu}_k,{x}_j(t),S)$. {\color{\change}We assume that the input signals are of bandwidth $W$ and periodic with period $T$. Also, $y(t)$ is observed over a time interval of length $T$.}
Many applications in communication and signal processing match this model, including radar imaging \cite{heckel2016super}, channel estimation \cite{shin2007blind}, microscopy \cite{mccutchen1967superresolution}, astronomy \cite{starck2002deconvolution} and localization problems \cite{xerri2002passive,amar2004direct}.  
In channel estimation, a wireless channel can be modeled as a LTV system with delay-Doppler shifts \cite{matz2011fundamentals}. A challenging problem in channel estimation is pilot contamination caused by sharing the non-orthogonal pilots among users. One way to deal with this problem is applying techniques that do not need any pilot signal, named blind methods \cite{shin2007blind}. As another example, we can mention spying radar where enemy collocated transmitters send unknown signals to some objects. {\color{\changess}The goal is to detect the intended objects and the transmitted signals from only one receive aperture. It is worth noting that due to the high complexity, it is not always possible to equip the receiver with multiple antennas in many applications among which we can mention unmanned aerial vehicle (UAV) and airborne receivers \cite{miso2010,mimoradarbook2018,sacco2018miso}.} In recent years, super-resolution methods based on convex optimization has attracted much attention \cite{candes2014towards,razavikia2019reconstruction,valiulahi2019two,safari2019off, fernandez2016super, tang2013compressed, duval2015exact, adcock2016generalized, bhaskar2013atomic, yang2015gridless, tan2014direction} due to their superior performance. This approach was first proposed by Candes and Fernandez-Granda in \cite{candes2014towards}. They used total variation norm for exact recovery of 1D spikes under a minimum separation condition with known system function and full measurements. Then,  \cite{ tang2013compressed} provided an atomic norm framework to estimate locations and amplitudes of a spike train in the frequency domain. In \cite{heckel2016super}, the authors apply atomic norm minimization to recover time and frequency shifts in radar application. They adapted super-resolution techniques of \cite{candes2014towards} to a single-input single-output (SISO) system with known {\color{\changes}band-limited} input signal.
The authors in \cite{suliman2018blind} investigate a similar problem, yet intend to estimate the time-frequency shifts in a blind way where the low-pass point spread function applied to the transmitter signal is unknown.
In this paper, we study the MISO system \eqref{A}. Besides its generality, this model matches the ``collocated'' transmitter scenario with ``many'' targets in MISO radar systems{\color{\changess} \cite{heckel2016mimo,sacco2018miso,mimoradarbook2018}} and differs from the previous models used in prior works. Here, shifts are independent from inputs and only depends on the system function. Moreover, unlike \cite{suliman2018blind}, we assume the input signals belong to ``different'' subspaces with disparate dimensions. While our model in \eqref{A} matches the well-studied model in MISO radar systems{\color{\changess} \cite{heckel2016mimo,sacco2018miso,mimoradarbook2018,miso2010}}, the strategy used in \cite{suliman2018blind,heckel2016super} is not be directly applicable in this setting. 
Generally, our goal in this paper is to find a strategy to detect the time-frequency shifts $(\widetilde{\tau}_k,\widetilde{\nu}_k),~ k=1,...,S$ as well as the transmitters' signals $x_j(t)$s in the MISO model \eqref{A}. 

{\color{\change}\textit{Notations}: We use boldface lowercase letters for column vectors and boldface uppercase letters for matrices. 
	We employ a two-dimensional index for vectors and matrices indicating how that vector or matrix is arranged i.e. $\left[\bm{a}\right]_{((m,l),1)} ,m,l=-N,\cdots, N$ means that $\bm{a}=[a_{(-N,-N)},a_{(-N,-N+1)},\cdots,a_{(N,N)}]^T$. $\text{conv}(C)$ is employed to denote the convex hull of the set $C$. The element-wise absolute value of the vector $\bm{x}$ is shown by $|\bm{x}|$. Moreover, $(\cdot)^T,(\cdot)^H, (\cdot)^*$ and $Tr(\cdot)$ stand for the transpose, hermitian, conjugate and trace, respectively while $Re \lbrace\cdot\rbrace$ denotes the real part of a complex scalar.
}
\vspace{-0.2cm}
%
\vspace{-0.1cm}
\section{System Model}\label{section.model}
We first sample $y(t)$ at rate $\frac{1}{W}$ based on 2WT-theorem  \cite{de2016limits} to collect $L:=WT$ samples, assumed to be an odd number. By applying the discrete Fourier transform (DFT) and the inverse DFT (IDFT) {\color{\change} to \eqref{A}} and defining the normalized parameters $\tau_k=\tfrac{\widetilde{\tau}_k}{T}$ and $\nu_k=\tfrac{\widetilde{\nu}_k}{W}$\footnote{Without loss of generality, we assume that $(\tau_k,\nu_k)\in[0,1]^2$.}, we obtain:
\begin{align}\label{B}
\begin{split}
&y(p):=y(p/W)=\\
&\tfrac{1}{L}\sum_{k=1}^Sb_k\sum_{j=0}^{N_I-1} \sum^N_{m,l=-N}x_j(l)e^{\tfrac{i2\pi m(p-l)}{L}}e^{i2\pi (p\nu_k-m\tau_k)}\\
&  \qquad p=-N,\ldots,N ,\quad L:=2N+1.
\end{split}
\end{align}
This is an under-determined linear system with $L$ equations and $N_IL+3S+1$ unknowns. To achieve a unique solution, we impose a subspace constraint \cite{ahmed2013blind, chi2016guaranteed, yang2016super, suliman2018blind, suliman2019exact}, which is common in a wide range of applications \cite{li2016identifiability}. {\color{\change}We assume that each input signal $\bm{x}_j:=[x_j(-N),\ldots,x_j(N)]^{\rm T}$ belongs to a known low-dimensional subspace with basis $\bm{D}_j\in\mathbb{C}^{L\times K_j}, K_j\ll L$, i.e.,
\begin{align}\label{C}
\bm{x}_j\!= \!\bm{D}_j \bm{h}_j ,~\bm{D}_j\! =\!\left[\bm{d}_{-N}^j,\ldots , \bm{d}_N^j\!\right]^ \textrm H\!\!\! \!\in\mathbb{C}^{L\times K_j} \!, \ \parallel \!\bm{h}_j\!\!\parallel_2=\!1,
\end{align}  
where ${\lbrace\bm{D}_j\rbrace}_{j=0}^{N_I-1}$ are known matrices whose columns span the signal space and ${\lbrace\bm{h}_j\rbrace}_{j=0}^{N_I-1}$ are unknown orientation vectors. This assumption has practical implications. For instance, in blind multi-user communication systems \cite{ahmed_blind}, the unknown message $\bm{h}_j$ of length $K_j$ corresponding to the $j$-th user is coded using a known tall coding matrix $\bm{D}_j$ and then the redundant coded message is transmitted. Thus,}
recovering $\bm{h}_j$ is equivalent to recovering $\bm{x}_j$ and incorporating the latter subspace constraints reduces the number of input variables from $\mathcal{O}(N_IL)$ to $\mathcal O(\sum_{j=0}^{N_I-1}K_j)$. We also assume that the columns of $\bm{D}_j^{\rm H}$ are sampled independently from a certain distribution satisfying isotropy and incoherence conditions (see \cite{Waking2016Wave} for more technical discussion).  Combining the definition of the Dirichlet kernel $D_N(t):=\tfrac{1}{L}\sum^N_{m=-N}e^{i2\pi tm}$
 and $\bm{x}_j(l)=\bm{d}_l^{j^H}\bm{h}_j$ in \eqref{B}, leads to (see Appendix A \cite{sayyari2020blind} for detailed derivation): 
\vspace{-0.1cm}
\begin{align}\label{E}
\begin{split}
&y(p)\!=\!\!\sum^S_{k=1}\!b_k\!\!\sum^{N_I-1}_{j=0}\!\!\sum^N_{m,l=-N}\!\!\!\!\!\!\!D_N(\tfrac{l}{L}-\tau_k)D_N(\tfrac{m}{L}-\nu_k) \bm{d}_{p-l}^{j^H}\bm{h}_j e^{\!\!\tfrac{i2\pi mp}{L}}\!\!\!\!.
\end{split}
\end{align}
Define vector $\bm{s}:=[\tau,\nu]^T$ and atoms $\bm{a}_j(\bm{s}_k)\in\mathbb{C}^{L^2}$ as:
\begin{align}\label{F}
\left[\bm{a}_j(\bm{s}_k)\right]_{((m,l),1)}=D_N\left(\tfrac{l}{L}-\tau_k\right)D_N\left(\tfrac{m}{L}-\nu_k\right).
\end{align}
The dictionary part $\widetilde{\bm{D}}_p^j \in \mathbb{C}^{L^2\times{K_j}}$ is defined as:
\begin{align}\label{G}
\left[\widetilde{\bm{D}}_p^j\right]_{((m,l),:)}=e^{\tfrac{i2\pi mp}{L} }\bm{d}_{p-l}^{j^H}\qquad p,l,m=-N,\ldots,N.
\end{align}
{\color{\change}Now, we substitute (\ref{F}) and (\ref{G}) in (\ref{E}) and exploit the lifting trick \cite{ahmed2013blind}, \cite{chi2016guaranteed} which is used for transforming nonlinear inverse problems into the problem of recovering a low-rank matrix from an under-determined system of linear equations.} This leads to:
\vspace{-0.05cm}
\begin{align}\label{H}
y(p)&\!\!=\!\!\sum^S_{k=1}\sum^{N_I-1}_{j=0}\!\!\! b_k\bm{a}_j^ \textrm H (\bm{s}_k)\widetilde{\bm{D}}^j_p\bm{h}_j\notag\\
&=\textrm {Tr}\!\left(\sum^{N_I-1}_{j=0} \!\!\!\widetilde{\bm{D}}_p^j\sum^S_{k=1}\!b_k\bm{h}_j\bm{a}_j(\bm{s}_k)^H\!\right)\!=\!\!\sum^{N_I-1}_{j=0}\!\!\langle{\bm{B}_j,\widetilde{\bm{D}}^{j  ^ \textrm H }_p}\rangle,
\end{align}
where $\bm{B}_j:=\sum^S_{k=1}\! b_k\bm{h}_j\bm{a}_j ^ \textrm H (\bm{s}_k)$. We define a linear operator $\chi:\oplus_{j=0}^{N_I-1}\mathbb{C}^{K_j\times L^2}\!\!\! \rightarrow \mathbb{C}^{L}$ that maps a {\color{\changes}matrix tuple} to a vector and the input {\color{\changes}matrix tuple}   $\mathcal B := (\bm{B}_j) ^{N_I-1} _{j=0}\in \oplus^{N_I-1}_{j=0}\mathbb{C}^{K_j\times L^2}$. Therefore, the observation vector $\bm{y} := [y(-N), \ldots, y(N)]^{\rm T}$ can be represented as $\bm{y}=\chi(\mathcal B)$.
%
%
In practice, the number of shifts $S$ is much smaller than the number of samples $L$. {\color{\change}Thus, each matrix $\bm{B}_j$ of the matrix tuple $\mathcal{B}$ can be described as a sparse linear combination of matrix elements taken from the atomic set
\begin{align}\label{J}
\mathcal{A}_j=\left\lbrace \bm{h}_j\bm{a}_j ^ \textrm H (\bm{s}): \bm{s}\in[0,1]^2,\parallel\! \bm{h}_j\!\parallel_2=1,\bm{h}_j\in\mathbb{C}^{K_j}\right\rbrace.
\end{align}
In fact, the atoms $\bm{h}_j\bm{a}_j ^ \textrm H (\bm{s})$ forms the building blocks of $\bm{B}_j$. To promote sparsity, we aim at minimizing the number of active atoms corresponding to each $\bm{B}_j$ subject to the constraints \eqref{H}. Thus, we have multiple non-convex objective functions. We instead{\color{\changes},} propose to minimize a relaxed function which is the sum of some convex functions:
%
\begin{align}\label{L}
\textrm P_1 \!\!: & \min_{\widetilde{\bm{B}}_j\in \mathbb{C}^{K_j\times L^2}}\!\!\sum^{N_I-1}_{j=0}\!\!\!\parallel\!\! \widetilde{\bm{B}}_j\!\!\parallel_{\mathcal{A}_j} ~~\text{s.t.} ~~ y(p)\!=\!\!\!\sum^{N_I-1}_{j=0}\!\!\langle\widetilde{\bm{B}}_j, \widetilde{ \bm{D}}^{j^\textrm H}_p\rangle,
\end{align}
where
\begin{align}\label{K}
\begin{split}
&\parallel \!\bm{B}_j\!\parallel_{\mathcal{A}_j}=\text{inf}\ \lbrace t>0:\bm{B}_j\in t\, \text{conv}(\mathcal{A}_j)\rbrace
\\ & =\!\!\underset {b_k\in\mathbb{C},\bm{s}\in[0,1]^2,\parallel \bm{h}_j\parallel_2=1}{\!\!\text{inf}} \lbrace \sum_{k} |b_k|\! :\bm{B}_j \!=\! \sum ^{}_{k} b_k \bm {h}_j \bm{a}_j ^ \textrm H (\bm{s}_k)  \rbrace,
\end{split}
\end{align} 
is called the atomic norm known as the best convex surrogate for the number of active atoms composing $\bm{B}_j$ \cite{chandrasekaran2012convex}.}
The optimization to calculate $\|\bm{B}_j\|_{\mathcal{A}_j}$ is over infinite dimensional variables and thus computationally intractable. To cope with this issue, we consider the dual problem and find a semidefinite relaxation for it in Subsection \ref{sec.dual_approach}. Aside from this, we propose a grid-based approach (dividing the region $[0,1]^2$ into grids) in Subsection \ref{sec.grid_based} to solve (\ref{L}) directly, leading to a reduced computational burden. 
\vspace{-0.1cm}
\subsection{Dual Approach}\label{sec.dual_approach}
The dual problem of (\ref{L}) is given by (the detailed derivation is included in Appendix B of \cite{sayyari2020blind}): 
\vspace{-0.15cm}
\begin{align}\label{I}
\begin{split}
&\underset{\bm{q}\in\mathbb{C}^{L}}{\text{max}}\quad\langle \bm{q},\bm{y}\rangle_\mathbb{R} ~~~\text{s.t.} ~~\parallel [\chi^*(\bm{q})]_j\parallel^d_{\mathcal{A}_j}\leqslant 1,
\end{split}
\end{align}
{\color{\change}where $\|.\|^d_{\mathcal{A}_j}$ denotes the atomic dual norm and} $\chi^*:\mathbb{C}^{L}\rightarrow \bigoplus^{N_I-1}_{j=0}\mathbb{C}^{K_j\times L^2}$ denotes the adjoint operator of $\chi$ such that $[\chi^*(\bm{q})]_j=\sum^N_{p=-N}\bm{q}_{p}\widetilde{\bm{D}}^{j^H}_p$ (see Appendix C of \cite{sayyari2020blind} for details). The dual norm in \eqref{I} is obtained as: 
\begin{align}\label{II}
\begin{split}
\parallel[\chi^*(\bm{q})]_j\parallel^d_{\mathcal{A}_j} & := \!\!\! \underset{\bm{s}\in[0,1]^2 \, , \, \parallel \bm{h}_j \parallel_2=1}{\text{sup}} \!\!\!  |\langle \bm{h}_j,[\chi^*(\bm{q})]_j\bm{a}_j(\bm{s})\rangle|\\
&=\underset{\bm{s}\in[0,1]^2}{\text{sup}}\parallel[\chi^*(\bm{q})]_j\bm{a}_j(\bm{s})\parallel_2\leqslant1.
\end{split}
\end{align}
Replacing (\ref{II}) in (\ref{I}) yields:
\begin{align}\label{N}
\begin{split}
\textrm P'_1 : \quad & \underset{\bm{q} \in \mathbb{C}^{L}} {\text{max}} \quad \langle \bm{q},\bm{y}\rangle_\mathbb{R} ~~\text{s.t.}~~\parallel\![\chi^*(\bm{q})]_j\bm{a}_j(\bm{s})\!\parallel_{2}\leqslant 1  ~~\bm{s}\in[0,1]^2
\end{split}
\end{align}
The primal convex problem (\ref{L}) has only equality constraint. Therefore, strong duality holds and in the optimal points we can claim $\sum^{N_I-1}_{j=0}\!\!\parallel\!\!\! \hat{\bm{B}}_j\!\!\!\parallel_{\mathcal{A}_j}=\langle \hat{\bm{q}},\hat{\bm{y}}\rangle_{\mathbb{R}}$.
The following theorem states conditions for optimality and uniqueness of the solution to this problem. Define the dual polynomial function as:
\begin{align}\label{O}
&\bm{f}_j(\bm{s}) := [\chi^* (\bm{q})] _j \bm{a}_j (\bm{s}) = \sum^{N}_{p=-N} \bm{q}_p {\widetilde{\bm{D}}^{j^ \textrm H}_p} \bm{a}_j(\bm{s}) \in \mathbb{C} ^{K_j}\nonumber\\
&=\sum^N_{p,m=-N}\!\!\!\!(\tfrac{1}{L} \ \bm{q}_{p}\sum^N_{l=-N}\!\!\!\bm{d}^j_l e^{\tfrac{i2\pi m(p-l)}{L}}) e^{-i2\pi (m\tau+p\nu)},
\end{align}
{\color{\changes}where the last equality is deduced from a similar discussion in \cite[Appendix C]{suliman2018blind}.}
\vspace{-0.1cm}
\begin{thm}\label{prop.optimality and uniqness of the solution} For the true support $\mathcal{S}=\lbrace \bm{s}_k\rbrace_{k=1}^{|S|}$ and the observation vector according to (\ref{H}), the matrix tuple $\hat {\mathcal B} = \mathcal B$ is the unique optimal solution of (\ref{L}) provided that the following conditions hold:
	\begin{enumerate} \item \label{eq.cond1} There exist 2D trigonometric vector polynomials (\ref{O}) with complex coefficients $\bm{q}=[q(-N),\ldots,q(N)]^{\rm T}$ such that:
		\vspace{-0.3cm}
		\begin{align}
		&\bm{f}_j(\bm{s}_k)=\text{sign}(b_k)\bm{h}_j,\forall \bm{s}_k \in \mathcal{S}\label{P}\\
		&{\color{\changess}{b}_k\in\mathbb{C}, \bm{h}_j\in\mathbb{C}^{K_j}, \parallel\!\!\bm{h}_j\!\!\parallel_2=\!1}\notag\\
		&\parallel \!\bm{f}_j(\bm{s}) \!\parallel_2 \textless 1, \ \forall \bm{s} \in [0,1]^2\ \setminus \ \mathcal{S}\label{Q}\\
		&\ j=0, ..., N_I-1 \notag
		\end{align}	
		\vspace{-0.5cm} 
		\item \label{eq.cond2} The sets $ \begin{Bmatrix} \bm{a}_j^ \textrm H (\bm{s}_k) \widetilde{\bm{D}}^j_{-N}, \ldots ,\bm{a}_j^ \textrm H (\bm{s}_k) \widetilde{\bm{D}}^j_{N}
		\end{Bmatrix}^{\!S}_{\!k=1}$\!\!\! are linearly independent.
	\end{enumerate}	
\end{thm}
\begin{proof}
	If $\bm{q}$ satisfies (\ref{P}) and (\ref{Q}), it will be in the feasible set of (\ref{N}) which is observed from \eqref{II} and \eqref{O}. Conversely, if $\bm{q}$ satisfies (\ref{P}) and (\ref{Q}), then $(\mathcal B$,$\bm{q})$ is a primal-dual optimal solution pair. To show this, by the definition of atomic norm in \eqref{K}, we have:
	\vspace{-0.3cm}
	\begin{align}
	\begin{split}
	&\langle \bm{q},\bm{y}\rangle_{\mathbb{R}}= \langle \chi^*(\bm{q}),\mathcal{B}\rangle_{\mathbb{R}} =\sum^{N_I-1}_{j=0}\langle [\chi^*(\bm{q})]_j,\bm{B}_j\rangle_{\mathbb{R}}\\
	&\!\!\stackrel{\eqref{H},\eqref{O}}{=}\!\!\sum^{N_I-1}_{j=0}\!\sum^S_{k=1} \!\!Re\lbrace b_k^*\langle \bm{h}_j,\bm{f}_j(\bm{s}_k)\rangle\rbrace\stackrel{\eqref{P}}{=}\!\sum^{N_I-1}_{j=0}\!\sum^S_{k=1} \!\!Re \lbrace b^*_k \text{sign}(b_k)\rbrace\\
	&=\sum^{N_I-1}_{j=0}\sum^S_{k=1} |b_k|\stackrel{\eqref{K}}{\geqslant} \sum^{N_I-1}_{j=0}\parallel \bm{B}_j \parallel _{\mathcal{A}_j}
	\end{split}\notag
	\end{align}
	\vspace{-0.1cm}
	On the other hand,
	\begin{align}
	\begin{split}
	&\langle \bm{q},\bm{y}\rangle_{\mathbb{R}}= \langle \chi^*(\bm{q}),\mathcal{B}\rangle_{\mathbb{R}}=\sum^{N_I-1}_{j=0}\langle [\chi^*(\bm{q})]_j,\bm{B}_j\rangle_{\mathbb{R}}\\
	&\leqslant \sum^{N_I-1}_{j=0} \parallel [\chi^*(\bm{q})]_j\parallel^d_{\mathcal{A}_j}\parallel \!\bm{B}_j \parallel_{\mathcal{A}_j}\stackrel{\eqref{P},\eqref{Q}}{\leqslant} \sum^{N_I-1}_{j=0} \parallel\! \!\bm{B}_j\! \parallel_{\mathcal{A}_j}.
	\end{split}\notag
	\end{align} 
	We conclude that $\sum^{N_I-1}_{j=0} \!\!\parallel {\!\bm{B}}_j\!\parallel_{\mathcal{A}_j}=\langle \bm{q},\bm{y}\rangle_{\mathbb{R}}$. Hence, $\mathcal{B}$ is the primal optimal and  $\bm{q}$ is the dual optimal solution. To check the uniqueness, assume that $\bar{\mathcal B}$ is another solution supported on $\bar{\mathcal{S}}\neq {\mathcal{S}}$ with $\bar{\bm{B}}_j := \sum_{ \bar{ \bm{s}}_k \in \bar{ \mathcal{S}}} \bar{b}_k \bar{\bm{h}}_j\bm{a}_j ^ \textrm H (\bar {\bm{s}_k})$, then:
	\begin{align}
	\begin{split}
	&\langle\bm{q},\bm{y}\rangle_{\mathbb{R}}=\langle\chi^*(\bm{q}),\bar{\mathcal{B}}\rangle_{\mathbb{R}}=\sum^{N_I-1}_{j=0} \langle[ \chi^* (\bm{q})]_j , \bar{\bm{B}}_j \rangle _{\mathbb{R}}\\
	& = \!\!\sum_{j}\sum_{\bar{\bm{s}}_k\in{\mathcal{S}}} \!\!\!Re\lbrace \bar{b}_k^*\langle \bar{\bm{h}}_j,\bm{f}_j(\bar{\bm{s}}_k)\rangle\rbrace\!+\!\!\!\sum_{j}\sum_{\bar{\bm{s}}_k\in\bar{{\mathcal{S}}} \setminus \mathcal{S}} \!\!\!\!\!\!Re\lbrace \bar{b}^*_k\langle\bar{\bm{h}}_j,\bm{f}_j(\bar{\bm{s}}_k)\rangle\rbrace \\
	& \stackrel{\eqref{Q}}{\textless} \sum_{j}\sum_{\bar{\bm{s}}_k\in{\mathcal{S}}}|\bar{b}_k|+\sum_{j}\sum_{\bar{\bm{s}}_k\in\bar{{\mathcal{S}}}\setminus \mathcal{S}} \!\!\!|\bar{b}_k|=\sum^{N_I-1}_{j=0}\parallel \bar{\bm{B}}_j \parallel_{\mathcal{A}_j}.\notag
	\end{split}
	\end{align}
	%
	Since the set of atoms and their shifts in $\mathcal{S}$ are linearly independent, having the same support means $\bar{\mathcal B}={\mathcal B}$. This result violates strong duality and therefore, $\mathcal B$ is the unique solution of (\ref{L}). 
\end{proof}
{\color{\change}Intuitively, the condition \eqref{eq.cond1} ensures that the primal and dual problems have zero duality gap achieved by the primal solution $\mathcal{B}$ and dual solution $\bm{q}$. The condition \eqref{eq.cond2} ensures that $\mathcal{B}$ is the unique optimal primal solution. A consequence of Theorem \ref{prop.optimality and uniqness of the solution} is that the dual solution $\bm{q}$ of \eqref{N} provides a way to determine the composing time and frequency shifts of $\mathcal{B}$. One could evaluate the vector-valued dual polynomials ${\lbrace\bm{f}_j(\bm{s})\rbrace}_{j=0}^{N_I-1}$ and localize the time and frequency shifts by identifying locations where $\|\bm{f}_j(\bm{s})\|_2$ achieves one.} 

However, the infinite number of constraints in the dual problem (\ref{N}) makes it computationally intractable. {\color{\change}To overcome this difficulty, 
we transform the $\ell_2$ ball of 2D trigonometric polynomial (the constraint of $\textrm P'_1$) into some linear matrix inequalities based on \cite{xu2014precise} or \cite[Chapter 9.4]{dumitrescu2017positive} and formulate \eqref{N} as the following semidefinite program (SDP):}
%
%
\vspace{-0.05cm}
\begin{align}\label{Pzeg1}
\begin{split}
\textrm P''_1:\quad &\underset{\bm{q},\bm{Q}_j\succeq 0}{\text{max}} \langle \bm{q},\bm{y}\rangle_\mathbb{R}\\
\textrm {s.t.} \, & \begin{pmatrix}
\bm{Q}_j&\widetilde{\bm{Q}}^H_j\\
\widetilde{\bm{Q}}_j & \bm{I}_{K_j}
\end{pmatrix}
\succeq 0 \, , \, \textrm {Tr} ((\bm{\theta}_m\otimes \bm{\theta}_l)\bm{Q}_j)=\delta_{m,l}
\end{split}
\end{align}
where  $\widetilde{\bm{Q}}_j \in \mathbb{C}^{K_j\times L^2}$ is defined as:
\vspace{-0.25cm}
\begin{align}
[\widetilde{\bm{Q}}_j]_{(i,(p,m))}=\tfrac{1}{L} \ \bm{q}_{p}\sum^N_{l=-N}\bm{d}^j_l e^{\tfrac{i2\pi m(p-l)}{L}}
\end{align}
and $\bm{\theta}_i$ is a Toeplitz matrix with ones on its $i$-th diagonal and zeros elsewhere.
{\color{\change} In general, we perform the following steps:
\begin{enumerate} 
	\item Solve (\ref{Pzeg1}) via one of the standard solvers e.g. SDPT3 of CVX package \cite{cvx} to find an estimate for the dual solution $\bm{q}$. 
	\item Determine the time-frequency shifts $(\hat \tau_k,\hat \nu_k,S)$ by computing the roots of the dual polynomials $ \bm{f}_j(\bm{s}_k)$ on the unit circle. This method leads to recovery with very high precision as shown in \cite{candes2014towards}. Another approach is to discretize the region $[0,1]^2$ on a fine grid up to a desired accuracy in order to identify the locutions $\hat{\bm{s}}_k$ where $\|\bm{f}_j(\hat{\bm{s}}_k)\|_2 \approx\!1$. We use this method in our numerical simulations in Section \ref{sec.simulations}. 
	\item Find the least squares solution of the following linear equations to estimate $\hat b_k \hat{\bm{h}}_j$: 
	\begin{align*}
	\sum^{N_I-1}_{j=0} \scriptstyle \begin{bmatrix} \scriptstyle\bm{a}_j ^ \textrm H (\bm{s}_1) \widetilde{\bm{D}}^j_{-N} \ \cdots \ \bm{a}_j ^ \textrm H (\bm{s}_S) \widetilde{\bm{D}}^j_{-N}\scriptstyle\\\vdots\\\scriptstyle\bm{a}_j ^ \textrm H (\bm{s}_1) \widetilde{\bm{D}}^j_{+N} \ \cdots \ \bm{a}_j ^ \textrm H (\bm{s}_S) \widetilde{\bm{D}}^j_{+N}
	\end{bmatrix} \begin{bmatrix} \scriptstyle b_1\bm{h}_j\\\vdots\\\scriptstyle b_S\bm{h}_j
	\end{bmatrix}
	=\begin{bmatrix} y(-N)\\\vdots\\ y(N)
	\end{bmatrix}.
	\end{align*} 
	\item By knowing the $\ell_2$ norm of the unknown message i.e. $\|\bm{h}_j\|_2$ (which is not a restrictive assumption e.g. in blind multi-user communications \cite{ahmed_blind}), we can estimate $|b_k|$ and thus $|\bm{h}_j|$.
\end{enumerate}
} 
\vspace{-0.3cm} 
\subsection{Grid-Based Approach} \label{sec.grid_based}
{\color{\changess}Program $\textrm P''_1$ in \eqref{Pzeg1} is a high-precision method that leads to an exact solution, but its computational complexity is high with variables of dimension ${\mathcal O}(L^4)$. Here, we provide a grid-based approach with reduced computational complexity in exchange for higher estimation error. In this approach, we suppose that the time-frequency shifts $(\tau_k,\nu_k)$ lie on a fine $(1/G,1/G)$-grid and solve the discrete version of the primal problem \eqref{L} with variables of dimension ${\mathcal O}(L^2)$.} Thus, the observation vector is expressed as:
\vspace{-0.05cm}
\begin{align*}
&y(p)\!=\!\tfrac{1}{L}\!\!\sum_{r,s=0}^{G-1} \!\!\!u_{r,s} \,e^{i2\pi p\tfrac {r}{G}}\!\!\sum_{j=0}^{N_I-1}\!\sum^N_{m,l=-N}\!\!\!\!\!\! x_j(l)e^{-i2\pi m\tfrac{s}{G}} e^{\tfrac{i2\pi m(p-l)}{L}},
\end{align*}
where $x_j(l)={\bm{d}_l^j}^ \textrm H \bm{h}_j$. Using the discrete version of the variables defined in the previous section, we get:
\vspace{-0.2cm}
\begin{align}\label{W}
y(p)\!=\!\!\sum^{G-1}_{r,s=0}\sum^{N_I-1}_{j=0}b'_{rs}\bm{a}_j(\bm{g})^H \widetilde{\bm{D}}^j_p\bm{h}_j=\sum^{N_I-1}_{j=0}\langle{\bm{B'}_j,\widetilde{\bm{D}}^{j^H}_p}\rangle,
\end{align}
where $\bm{g}=[r/G,s/G]$, $\bm{b'}\in \mathbb{C}^{G^2}$ is a sparse vector such that $\bm{b'}_{((r,s),1)}:=b'_{rs}$, and $\bm{B'}_j:=\sum^{G-1}_{r,s=0}b'_{rs}\bm{h}_j\bm{a}_j(\bm{g})^H$.

\vspace{0.15cm}
Now, we propose the following optimization program to recover the time-frequency shifts:
\vspace{-0.2cm}
\begin{align}\label{P_2}
\textrm P_2 \!:\!\min_{\bm{B'}_j\in \mathbb{C}^{K\times G^2}}\!\!\sum^{N_I-1}_{j=0}\!\!\!\parallel\! \bm{B'}_j\!\parallel_{*}~~\text{s.t.}~~y(p)\!=\!\!\sum^{N_I-1}_{j=0} \langle{\bm{B'}_j,\widetilde{\bm{D}}^{j^H}_p}\rangle
\end{align}
where $\parallel \cdot \parallel_{*}$ is the nuclear norm of a matrix which is regarded as an alternative for the atomic norm in the discrete setting.

	
{\color{\change}In the grid-based approach, we are able to solve (\ref{P_2}) which is the discrete version of (\ref{L}) directly. Then, to recover the location of shifts, we compute the correlation between the estimated ${\lbrace\bm{B'}_j\rbrace}_{j=0}^{N_I-1}$ and the atoms and choose the locations with highest correlation. The remaining unknowns are obtained similar to Subsection \ref{sec.dual_approach}.}
\vspace{-0.1cm}
\section{Simulation Results} \label{sec.simulations}
In this section, we provide some experiments to confirm the accuracy of our proposed methods using CVX toolbox and SDPT3 package. Note that here we consider a minimum separation condition between time-frequency shift pairs {\color{\changes}similar to} \cite{heckel2016super}. In the first experiment, we have two input signals and a system with the time-frequency shift pair $(0.24,0.52)$, generated uniformly in the interval $[0,1]$. We set $N=7$, $K_1=K_2=1$, $S=1$ and $N_I=\!2$. The entries of $\bm{D}_j$, $\bm{h}_j$ and $b_k$ are generated from complex standard normal distribution subject to $\parallel\! \bm{h}_j\! \parallel_2 =\!1$ and $|b_k|=1$. Fig.\ref{fig.dualpoly} shows the dual polynomials in (\ref{O}) which achieve $1$ in the locations of true shift pairs. The top left image of Fig.\ref{fig.inputs+grid} illustrates perfect matching of the magnitude corresponding to the true and estimated input signals $x_j(l)$. Note that here we are not able to recover the phases of $\hat b_k$, instead, $\hat b_k \hat{\bm{h}}_j$ can be uniquely estimated. Since $|b_k|=1$, the least square algorithm estimates $|\hat x_j|$. 
For the grid-based approach of \eqref{P_2}, the shift pair is $(0.92,0.67)$ while $N_I\!=\!2$, $N\!=\!6$, $K_1\!=\!2$, $K_2\!=\!1$ and $S\!=\!1$. We define the super-resolution factor $\textrm {SRF} \!:=\! \tfrac{G}{L}$ and $\textrm {Error}\!:=\!\tfrac{L}{S} (\sum_{k=1}^S\sqrt{(\hat{\tau}_k\! -\! \tau_k)^2\!+\!(\hat{\nu}_k - \nu_k)^2})$. The right block of Fig.\ref{fig.inputs+grid} indicates that the error is inversely related to SRF. The bottom left image of Fig.\ref{fig.inputs+grid} shows good matching of the true and estimated signals when $SRF\!=\!20$.
{\color{\changess}In the last experiment, we evaluate the computational complexity of our methods by CPU time. We set the parameters as in the first experiment. We implement both $\textrm P''_1$ and $\textrm{P}_2$ using CVX package \cite{cvx}. While the CPU time of the dual approach $\textrm{ P}''_1$ is $5019$ seconds, the grid-based method ${\textrm P}_2$ is implemented in $576$ seconds. Our simulation is implemented using MATLAB 2016 on a laptop computer with specifications 2.0 GHz, Core i7, 8 GB RAM.}
\vspace{-0.2cm}
\begin{figure}[t!]\vspace{-0.5cm}
	\includegraphics[scale=.19]{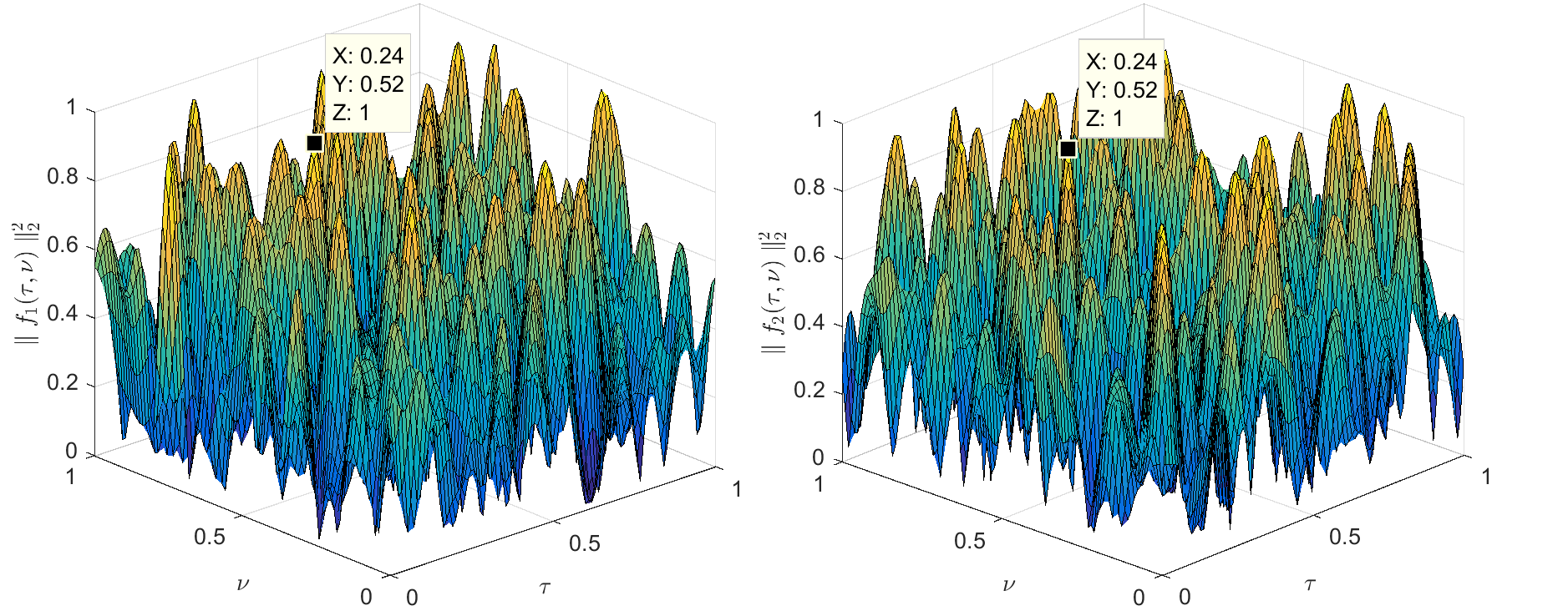}
	\caption{The Euclidean norm of the dual polynomial vectors \eqref{O} corresponding to each input. {\color{\changess}The marker shows the estimated time-frequency shifts corresponding to the locations where the dual polynomial has unit $\ell_2$ norm. The obtained location is the same as the true time-freq. shift pair $(\tau=0.24,\nu=0.52)$.}}
	\label{fig.dualpoly}
\end{figure}
\vspace{-0.1cm}
\begin{figure}[t!]\vspace{-0.2cm}
	\hspace{-0.1cm}
	\includegraphics[scale=.19]{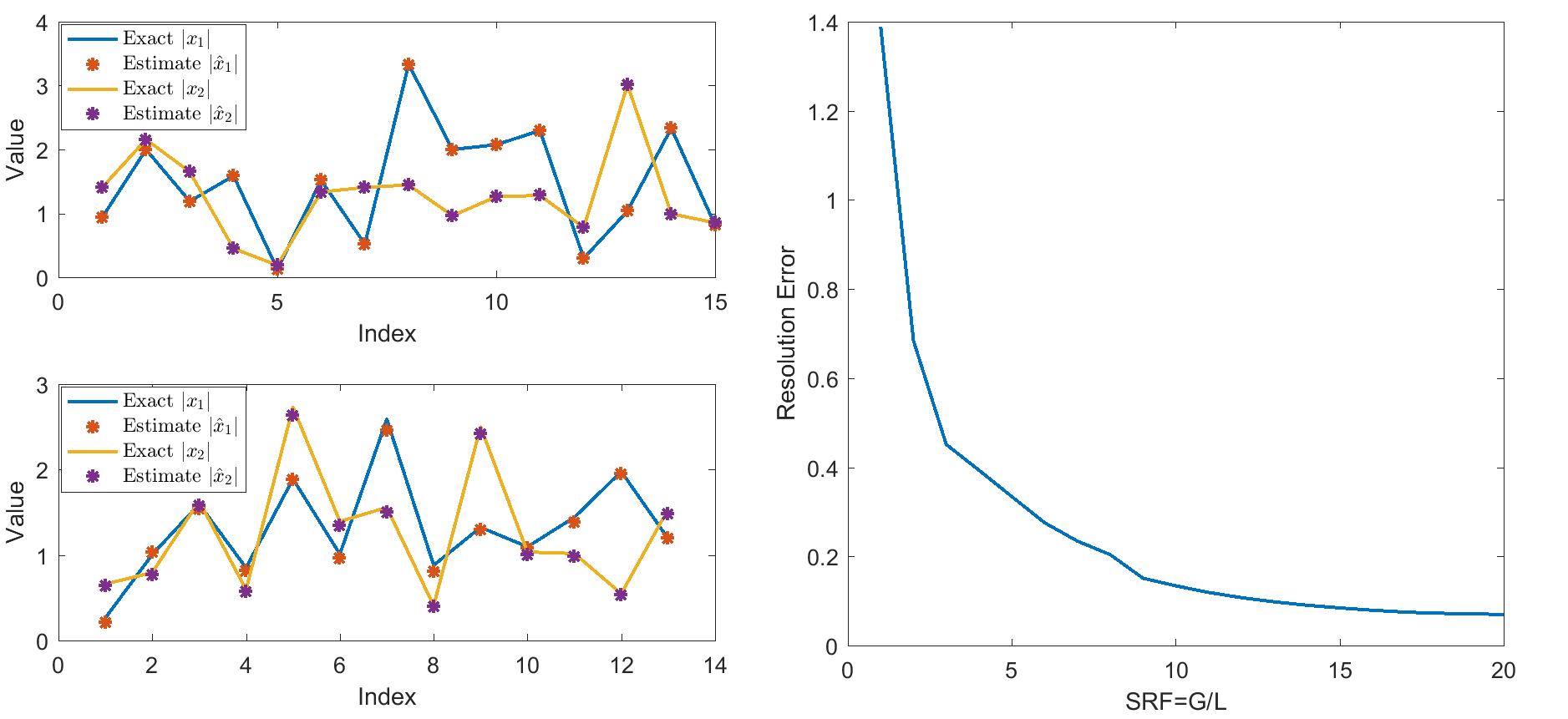}
	\caption{Left block: The estimated and true input signals match in the dual approach (upper graph) and grid-based approach (lower graph). Right block: Estimation error is inversely related to the super resolution factor (SRF) in the grid-based approach.} 
	\label{fig.inputs+grid}
\end{figure}
%
\vspace{0.3cm}
{\color{\change}\section{CONCLUSION}
	In this work, we developed a mathematical model for 2D blind super-resolution in MISO LTV systems. The output signal in this model is a superposition of time and frequency shifted versions of multiple independent inputs. Our aim was to recover the time-frequency shift pairs and the inputs from only one output vector. We formulated the model as a convex optimization problem under different subspace assumptions for different unknown inputs. Then, a dual optimization and a grid-based approach were proposed for obtaining the off-grid time-frequency shifts and the input signals up to a scaling factor. Also, numerical simulations were provided to verify our proposed approach.
}

\newpage
\appendices
\section{equivalence between (\ref{B}) and (\ref{E})}\vspace{-0.6cm}   \label{conversion to matrix-vector form}
\begin{align}
\begin{split}
&y(p)=
\tfrac{1}{L}\sum_{k=1}^Sb_ke^{i2\pi \nu_k p}\sum_{j=0}^{N_I-1}(\sum^N_{m,l=-N}x_j(l)e^{i2\pi (\tfrac{p-l}{L}-\tau_k)m})\\
& =\tfrac{1}{L}\sum^S_{k=1}b_ke^{i2\pi \nu_k p}\sum^{N_I-1}_{j=0}(\sum_{u=p-N}^{p+N}\sum^N_{k=-N}x_j(p-u)e^{i2\pi (\tfrac{u}{L}-\tau_k)m})\\
&  \stackrel{\text{(i)}}{=}\sum^S_{k=1}b_k\sum^{N_I-1}_{j=0}e^{i2\pi \nu_k p}(\tfrac{1}{L}\sum^N_{m,l=-N}x_j(p-l)e^{i2\pi (\tfrac{l}{L}-\tau_k)m})\\
&\stackrel{\text{(ii)}}{=}\sum^S_{k=1}b_k\sum^{N_I-1}_{j=0}e^{i2\pi \nu_k p}(\sum^{N}_{l=-N}x_j(p-l)D_N(\tfrac{l}{L}-\tau_k))\\
&\stackrel{\text{(iii)}}{=}\sum^S_{k=1}b_k \sum^N_{m,l=-N}D_N(\tfrac{m}{L}-\nu_k)D_N(\tfrac{l}{L}-\tau_k) x_j(p-l)e^{\tfrac{i2\pi ml}{L}} \notag
\end{split}
\end{align}
(i) is based on the periodicity property of $x_j(l)$ while (ii) is based on the definition {\color{\changess}of the Dirichlet kernel.} (iii) is deduced from the following fact:
\vspace{-0.3cm} 
\begin{align} 
\sum_{m=-N}^N D_N\left(\tfrac{m}{L}-\nu_k\right)e^{\tfrac{i2\pi ml}{L}}=e^{i2\pi \nu_kp}\notag
\end{align} 
\vspace{-0.6cm}
\section{Proof of the dual problem (\ref{I})} \vspace{-0.1cm} 
\label{Proof of the dual problem}
The Lagrangian function of (\ref{L}) is equal to: 
\begin{align}
\begin{split}
\mathcal{L}(\mathcal B, \bm{q}) & = \sum_j \parallel\! \bm{B}_j\!\parallel_{\mathcal{A}_j}+\langle \bm{q},\bm{y} -\chi(\mathcal B)\rangle
\end{split}\notag
\end{align}
\vspace{-0.2cm}
in which:
\begin{align}
\begin{split}
\langle \bm{q},\chi(\mathcal B)\rangle & = \langle \chi^* (\bm{q}), \mathcal B \rangle = \sum_j \langle[\chi^*(\bm{q})]_j,\bm{B}_j\rangle\\
&\leqslant \sum_j\parallel \!\bm{B}_j\! \parallel_{\mathcal{A}_j} \parallel [\chi^*(\bm{q})]_j\!\parallel^d_{\mathcal{A}_j}
\end{split}\notag
\end{align}
and Therefore,
\begin{align}
\begin{split}
\mathcal{L}(\mathcal B, \bm{q}) \geqslant \sum_j \parallel\! \bm{B}_j\! \parallel_{\mathcal{A}_j} (1-\parallel\! [\chi^*(\bm{q})]_j\!\parallel^d_{\mathcal{A}_j}) +\langle \bm{q},\bm{y} \rangle
\end{split}\notag
\end{align}
We obtain the dual function by minimizing over $\mathcal B$:
\begin{align}
\begin{split}
H(\bm{q}) = \!\!\!\!\!\!  \underset{{\mathcal B \in \oplus_j \mathbb{C}^{K_j\times L^2}}}{\inf} \!\!\!\!\!\! \mathcal{L}(\mathcal B, \bm{q}) 
=\left \lbrace
\begin{array}{ll}
\langle \bm{q},\bm{y}\rangle  & \parallel [\chi^*(\bm{q})]_j \!\! \parallel^d_{\mathcal{A}_j}\leqslant 1\\
+\infty & \textrm{otherwise}
\end{array}
\right .
\end{split}\notag 
\end{align}
So, the dual problem can be written as (\ref{I}).
\vspace{-0.2cm} 
\section{Adjoint Operator $\chi^*$}
\label{adjoint operator}
Beginning from (\ref{N}), we have:
\vspace{-0.1cm}
\begin{equation}
\label{eq:adj:qy}
\langle \bm{q},\bm{y}\rangle_{\mathbb{R}}= \langle\bm{q},\chi(\mathcal B)\rangle_{\mathbb{R}} =  \sum_{p,j} \langle \bm{B}_j,\widetilde{\bm{D}}^{j^ \textrm H}_p\rangle\,\bm{q}_{p} 
\end{equation}
\vspace{-0.2cm} 
We also have:
\begin{equation}
\label{eq:adj:qy2}
\langle \bm{q},\bm{y}\rangle_{\mathbb{R}}=\langle \chi^*(\bm{q}),\mathcal{B}\rangle_{\mathbb{R}}=\sum_j\langle \bm{B}_j,[\chi^*(\bm{q})]_j\rangle_{\mathbb{R}}
\end{equation}
From \eqref{eq:adj:qy} and \eqref{eq:adj:qy2} we can deduce:
\begin{align}
\begin{split}
&\sum_j\langle \bm{B}_j,\sum_p\bm{q}_{p}\widetilde{\bm{D}}^{j^ \textrm H}_p\rangle=\sum_j \langle \bm{B}_j,[\chi^*(\bm{q})]_j\rangle_{\mathbb{R}}
\end{split} \notag
\end{align}
Therefore:
\vspace{-0.2cm}
\begin{align}
[\chi^*(\bm{q})]_j=\sum_p\bm{q}_{p}\widetilde{\bm{D}}^{j^ \textrm H}_p \notag
\end{align}

\bibliographystyle{ieeetr}
\bibliography{mypaperbibenew}

\end{document}